\documentclass[10pt,twocolumn]{article}

\newcommand{\papertitle}{CONTINUITY: Security-Context Contracts for Composable LLM Agent Controls}
\newcommand{\paperauthor}{Chris Zheng; Geng Yang}
\newcommand{\authorone}{Chris Zheng}
\newcommand{\authortitleone}{Co-Founder \& Chief Researcher}
\newcommand{\authoremailone}{chris@zast.ai}
\newcommand{\authortwo}{Geng Yang}
\newcommand{\authortitletwo}{Co-Founder \& CEO}
\newcommand{\authoremailtwo}{geng@zast.ai}
\newcommand{\authoraffiliation}{ZAST.AI}
\newcommand{\paperdate}{September 5, 2026}
\newcommand{\artifacturl}{https://github.com/zast-ai/continuity}
\newcommand{\artifactdisplay}{github.com/zast-ai/continuity}
\newcommand{\paperlicense}{CC BY 4.0}
\newcommand{\paperlicenseurl}{https://creativecommons.org/licenses/by/4.0/}
\newcommand{\codelicense}{MIT}
\newcommand{\codecopyrightholders}{Chris Zheng, Geng Yang, and ZAST.AI}

\usepackage[letterpaper,margin=0.70in,columnsep=0.24in]{geometry}
\usepackage[T1]{fontenc}
\usepackage[utf8]{inputenc}
\usepackage{lmodern}
\usepackage{microtype}
\usepackage{amsmath,amssymb,amsthm,mathtools}
\usepackage{booktabs,array,multirow,tabularx}
\usepackage{graphicx}
\usepackage{xcolor}
\usepackage{enumitem}
\usepackage{listings}
\usepackage{url}
\usepackage[backend=biber,style=numeric-comp,sorting=none,maxbibnames=99]{biblatex}
\usepackage{balance}
\usepackage{needspace}
\usepackage{float}
\usepackage{tikz}
\usepackage{hyperref}
\usepackage[nameinlink,noabbrev]{cleveref}
\usetikzlibrary{arrows.meta,positioning,fit,shapes.geometric,calc}

\definecolor{codegray}{gray}{0.96}
\newtheorem{definition}{Definition}
\newtheorem{theorem}{Theorem}
\newtheorem{lemma}{Lemma}
\newtheorem{proposition}{Proposition}

\newcommand{\system}{\textup{\textsc{Continuity}}}
\newcommand{\eci}{\textup{\textsc{ECI}}}
\newcommand{\fieldset}{\mathcal{F}}
\newcommand{\contract}{\mathcal{C}}

\newcommand{\canon}{\mathsf{Canon}}
\newcommand{\effects}{\mathsf{Effects}}
\newcommand{\verify}{\mathsf{Verify}}
\newcommand{\realize}{\mathsf{Realize}}
\newcommand{\code}[1]{\texttt{#1}}

\hypersetup{
  pdftitle={\papertitle},
  pdfauthor={\paperauthor},
  pdfsubject={Research artifact: https://github.com/zast-ai/continuity; paper license: CC BY 4.0},
  colorlinks=true,
  linkcolor=blue!55!black,
  citecolor=blue!55!black,
  urlcolor=blue!55!black
}

\title{\papertitle}
\author{
  \begin{tabular}{c}
    \authorone\\[-0.10em]
    {\small \authortitleone}\\[-0.10em]
    {\small \authoraffiliation}\\[-0.10em]
    {\small\texttt{\authoremailone}}
  \end{tabular}
  \and
  \begin{tabular}{c}
    \authortwo\\[-0.10em]
    {\small \authortitletwo}\\[-0.10em]
    {\small \authoraffiliation}\\[-0.10em]
    {\small\texttt{\authoremailtwo}}
  \end{tabular}
}
\date{
  \paperdate\\[0.30em]
  {\small Artifact: \href{\artifacturl}{\nolinkurl{\artifactdisplay}} \quad | \quad License: \href{\paperlicenseurl}{\paperlicense}}
}

\begin{document}
\maketitle

\begin{abstract}
Tool-using large language model agents increasingly combine provenance tracking, task-scoped authorization, policy gateways, protocol adapters, and effect-bound execution. These controls are commonly specified and tested in isolation. Their composition can nevertheless fail when a boundary drops a source label, accepts a self-declared authority root, widens a delegation, changes an approved field without a valid transformation relation, or executes a stale or replayed permit. We call this failure class \emph{security-context discontinuity}.

We define \emph{end-to-end consequence integrity} (\eci): every externally realized effect must have a verifiable witness connecting the exact effect to an authenticated chain origin, principal, task, field-level provenance, root grant, component contracts, current policy, finality sink, and single-use execution state. We introduce an assume--guarantee contract model and a reference system, \system{}, that enforces signed root grants; role-bound component identities; RFC~6901-style leaf paths; signed provenance and context manifests; bounded, source- and value-bound typed releases; independently verifiable transformation witnesses; and subject-, action-, policy-, revocation-, and replay-bound finality permits.

The 1.5-KLOC trusted-core prototype is evaluated with a deterministic conformance suite comprising 32 cross-layer fault classes, four domains, and 20 parameterized instances per fault--domain pair, plus 700 benign and 200 ambiguous tasks. Across 3,460 scenarios and 24,220 system--scenario runs, \system{} contains all 128 fault--domain classes and commits no harmful effect in 2,560 attack instances, while automatically completing all benign tasks and escalating all ambiguous tasks. The strongest incomplete reference configuration commits a harmful effect in 65.6\% of attack instances. Targeted ablations reopen 4--24 fault--domain classes depending on the omitted invariant. Median proof verification is 4.21~ms and median end-to-end transition production, verification, permit issuance, and finality is 7.17~ms on the recorded host. These are deterministic conformance results under an explicit trusted-computing-base model, not an estimate of real-world attack probability or a claim to solve semantic correctness.
\end{abstract}

\section{Introduction}
\label{sec:introduction}

Large language model (LLM) agents increasingly read attacker-writable content, synthesize plans, call external tools, delegate work, and create effects that outlive the model invocation. Indirect prompt injection demonstrates that data retrieved from a web page, email, document, or tool response can influence an agent toward attacker-chosen actions~\cite{greshake2023indirect,zhan2024injecagent,debenedetti2024agentdojo}. System-level defenses therefore separate untrusted data from control, constrain tool privileges, track causal support, or interpose policy at execution~\cite{debenedetti2025camel,kim2025pfi,shi2025progent,weng2026argus,santosgrueiro2026cxi}.

A deployed agent rarely uses only one control. A typical path contains an ingress authenticator, provenance service, memory layer, policy gateway, protocol adapter, tool server, and final effect sink. Each component may be locally reasonable while the end-to-end path is insecure. Consider an email action whose destination was derived from an untrusted invoice. A provenance layer labels the destination and a policy gateway requires a validated release. An adapter later serializes the action but drops the source binding, rewrites a logical alias, or binds the permit to a different representation. A final sink sees a valid signature and an allowed tool name, yet cannot determine whether the concrete destination is the one the gateway approved. The failure is not necessarily in prompt classification or in either endpoint; it is in the composition.

We call this class of failure \emph{security-context discontinuity}: a security-relevant fact is lost, weakened, reinterpreted, or rebound between controls. Discontinuity has several forms. A chain may begin at a self-declared authority root; a component signing key may be accepted for a role it was never authorized to perform; a release may name a field but fail to bind the source value or validity interval; an adapter may perform an allowed-looking transformation without proving the declared relation; or a finality sink may omit subject, policy, revocation, or replay checks. Cryptographic signatures alone do not prevent these failures. A signature authenticates who emitted an object, not whether the emitter was authorized for that stage or whether its transformation preserved the upstream security contract.

This paper asks:

\begin{quote}
\emph{Under what explicit conditions do independently useful agent-security controls compose into an end-to-end consequence boundary?}
\end{quote}

We model the LLM planner as adversarial: it may propose any action, parameters, destination, tool, or natural-language justification. The trusted computing base (TCB) consists of explicitly configured root issuers, provenance and release validators, component-role bindings, the deterministic verifier, a key service, and mediated finality sinks. Our goal is not to prove that the model reasons correctly. It is to ensure that an external effect cannot be committed merely because the model proposed it.

We define \emph{end-to-end consequence integrity} (\eci). Informally, every realized effect must carry a verifiable witness from an authenticated root grant through each authorized transition to a current, single-use, effect-bound permit. The witness includes field-level source commitments, component identities and contracts, any typed release or semantic transformation evidence, and finality state. We then express each control as an assume--guarantee contract:
\[
\contract_i=(A_i,G_i,P_i,M_i),
\]
where $A_i$ are input assumptions, $G_i$ are output guarantees, $P_i$ are fields that must be preserved, and $M_i$ are explicit transformation relations. Safe composition requires that upstream guarantees discharge downstream assumptions and that every changed security field is either preserved or justified by an independently checked relation witness. This follows the broader tradition of compositional specifications~\cite{abadi1993composing}, but instantiates the interface around agent authority, provenance, action semantics, and finality.

We implement these ideas in \system{}. The prototype uses signed root grants, trusted ingress identities, fixed component-role-contract topology, JSON Pointer leaf paths~\cite{rfc6901}, signed provenance/context manifests, signed typed-release credentials, transformation witnesses, and one-shot finality permits. The adapter performs a real security-relevant transformation in every normal scenario: a logical destination alias is resolved to a canonical sink address. The transformation is accepted only when a trusted directory witness binds the before value, after value, relation, component, contract, task, and expiry. Thus, the evaluation does not obtain security by simply forbidding all changes.

We make five contributions:

\begin{enumerate}[leftmargin=1.2em,itemsep=0.25em]
  \item \textbf{Problem and taxonomy.} We identify security-context discontinuity as a distinct composition failure and organize 32 fault classes spanning root trust, provenance and release, component topology, semantic transformation, policy freshness, finality, and replay.
  \item \textbf{Formal model.} We define \eci, role-bound transition witnesses, and sufficient composition conditions: authenticated roots, field-level continuity, authority non-amplification, validated transformations, current policy, exact effect binding, and complete mediation.
  \item \textbf{Reference design.} We present proof-carrying context transitions with signed manifests, bounded typed releases, semantic relation witnesses, and state-revalidated finality permits.
  \item \textbf{Executable artifact.} We provide a dependency-light Python implementation, 30 regression tests, deterministic fault injection, seven reference configurations, targeted ablations, raw results, and reproducibility scripts.
  \item \textbf{Evaluation.} Across 128 fault--domain classes and 2,560 parameterized attacks, the full configuration commits no harmful effect, while incomplete compositions fail on 4--128 classes. All 700 benign tasks, including 300 bounded external-data releases and a legitimate destination transformation, complete automatically; all 200 unreleased ambiguous tasks escalate.
\end{enumerate}

The result is deliberately narrower than ``solving prompt injection.'' A malicious planner may still produce harmful text, bad recommendations, or denial of service. \system{} instead establishes a testable structural claim: if the declared TCB and composition conditions hold, a planner-controlled proposal cannot directly cross the finality boundary without a valid end-to-end witness.

\section{Background and Motivation}
\label{sec:background}

\subsection{From prompt compromise to external consequence}

Prompt injection is dangerous when model output is connected to an effectful interface. Benchmarks such as InjecAgent and AgentDojo measure whether attacker-controlled observations can redirect tool-using agents~\cite{zhan2024injecagent,debenedetti2024agentdojo}. A model-level compromise and a system-level security failure are related but not identical. The model may propose a forbidden action while a deterministic execution boundary rejects it; conversely, the model may propose an apparently benign action whose parameters are later changed or misinterpreted by an adapter.

We therefore separate three events:

\begin{enumerate}[leftmargin=1.2em,itemsep=0.2em]
  \item \emph{proposal}: the planner emits an action candidate;
  \item \emph{admission}: a security control authorizes a canonical action;
  \item \emph{effect}: a sink commits a state transition in the outside world.
\end{enumerate}

Our primary metric is harmful effect realization, not whether the planner generated the proposal. The planner is assumed compromised in every attack scenario.

\subsection{Useful controls and their interfaces}

Current defenses contribute complementary properties. CaMeL extracts trusted control and data flow and uses capabilities to constrain exfiltration~\cite{debenedetti2025camel}. Prompt Flow Integrity combines isolation, secure processing of untrusted data, and privilege guardrails~\cite{kim2025pfi}. IsolateGPT separates execution domains~\cite{wu2024isolategpt}; Progent expresses programmable privileges~\cite{shi2025progent}; formal-policy work compiles and enforces runtime policies~\cite{palumbo2026formal}; and ARGUS audits causal provenance from runtime evidence to proposed actions~\cite{weng2026argus}. CXI binds protected fields, typed releases, exact effects, and invocation authority at an execution boundary~\cite{santosgrueiro2026cxi}.

These properties are not interchangeable. A provenance tracker does not itself mediate all sinks. A capability authorizes a subject but may not bind the concrete destination. A finality token may bind an exact action but be issued after a compromised adapter has already discarded the provenance needed for policy. A runtime hook provides an interception point but does not specify which security context must survive across the hook. OWASP's Agent Control Standard focuses on runtime control hooks and policy decisions~\cite{owasp2026acs}; the Model Context Protocol specifies tool discovery and invocation and separately documents security considerations~\cite{mcp2026tools,mcp2026security}. \system{} is designed as a context and contract layer that can be carried across such interfaces; the artifact itself is protocol-neutral and does not claim production integration with either standard.

\subsection{Why signatures are insufficient}

Suppose an adapter signs an output action. Signature verification establishes that the adapter produced that action. It does not establish:

\begin{itemize}[leftmargin=1.2em,itemsep=0.15em]
  \item that the adapter was authorized for this stage or contract;
  \item that the chain began at an accepted authority root;
  \item that an input field and output field denote the same value;
  \item that an altered destination satisfies an approved transformation;
  \item that a release applies to this source value and task;
  \item that the permit is still current and unused at the sink.
\end{itemize}

This distinction parallels software-supply-chain systems that authenticate a sequence of authorized transformations rather than merely signing the final artifact~\cite{torresarias2019intoto,samuel2010tuf}. \system{} similarly authenticates the transition relation, not only each endpoint.

\subsection{Security principles}

The design instantiates four established principles. First, \emph{complete mediation}: every path in an effect-equivalence class must cross a compatible sink~\cite{saltzer1975protection}. Second, \emph{least authority}: a component may exercise only a task- and resource-scoped grant, consistent with zero-trust deployment guidance~\cite{nist2020zerotrust}. Third, \emph{canonical binding}: signatures and permits cover a deterministic representation; the prototype uses a restricted deterministic JSON encoding and identifies RFC~8785 as the interoperability target~\cite{rfc8785}. Fourth, \emph{freshness and audience restriction}: credentials are short-lived, scoped to a particular sink, and replay-protected, consistent with modern authorization guidance~\cite{rfc9700}.

The distinctive question in this work is not whether each principle is valuable, but what must be carried and checked so that the principles continue to hold after several independently implemented controls transform the request.

\section{System and Threat Model}
\label{sec:model}

\subsection{Entities and execution path}

A task traverses the following abstract path:
\[
\textsf{Ingress}\rightarrow C_1\rightarrow\cdots\rightarrow C_n
\rightarrow\textsf{Verifier}\rightarrow\textsf{FinalitySink}.
\]
The ingress produces the first signed envelope $E_0$. Each component $C_i$ produces a new envelope $E_i$ and a signed transition receipt $\rho_i$. The verifier admits or rejects the complete bundle. On admission it creates an execution permit $\pi$, and a finality sink consumes $\pi$ to commit an effect and produce an outcome receipt $\omega$.

The LLM planner may run within or alongside one or more components, but it is not part of the TCB. Its natural-language explanation is never treated as authority. The action it proposes is data until a deterministic verifier and sink admit it.

\subsection{Adversary}

The adversary controls all natural-language and structured content supplied by external web pages, emails, documents, tool responses, and remote agents. It may cause the planner to emit arbitrary operations, arguments, resources, destinations, tool identifiers, delegation requests, retries, or explanations. It may exploit a faulty or compromised non-root pipeline component to sign an invalid envelope or receipt with that component's key. It may reorder or truncate chain elements, substitute signed objects, replay permits, mutate an action after authorization, invoke an alternate effect path, or present stale policy and revocation state.

The adversary cannot break SHA-256 collision resistance or forge an Ed25519 signature under an uncompromised key~\cite{rfc8032}. It cannot directly modify verifier or finality-sink code, bypass all operating-system isolation around the TCB, or compromise every mandatory effect path. These are explicit assumptions, not properties proved by the protocol.

\subsection{Trusted computing base}

The TCB contains:

\begin{itemize}[leftmargin=1.2em,itemsep=0.15em]
  \item the deployment policy naming trusted ingress, root-grant, provenance, context, release, and transformation issuers;
  \item the mapping from pipeline stage to authorized component identity and contract;
  \item trusted tool-manifest digests and the current policy state;
  \item deterministic field predicates, transformation relations, and canonicalization;
  \item the continuity verifier, permit issuer, revocation state, nonce/idempotency ledger, and mediated finality sinks;
  \item cryptographic keys and their provisioning, rotation, and revocation mechanisms.
\end{itemize}

A key is trusted only for its configured role. Possession of a registered component key does not authorize that key to originate a chain, issue a root grant, validate provenance, or sign a release. Conversely, compromise of a key accepted as a chain origin or root authority defeats the guarantees rooted in that key. Monotonicity is relative to an authenticated origin; it cannot repair a maliciously authorized root. Compromise of a stage key is narrower: the attacker can impersonate that stage, but a valid chain must still satisfy role, predecessor, contract, transformation, root, finality, and other independently enforced checks. The security analysis in \cref{sec:security} makes these dependencies explicit.

\subsection{Effect model}

An action $a$ is a structured object containing an operation, tool and server identity, manifest digest, function, resource, destination, parameters, effect class, and data classification. A sink interprets $a$ and may realize an external effect $e$. Two syntactically different actions can be effect-equivalent:
\[
a\equiv_e a'\quad\text{iff}\quad \effects(a)=\effects(a').
\]
For example, email, HTTP, browser, shell, and remote-agent interfaces may all exfiltrate the same data. Complete mediation must cover the effect class, not only one tool name.

\subsection{Security objective and non-goals}

Our objective is structural consequence integrity: a harmful proposal must not become an externally committed effect without a complete, current authorization witness. We do not attempt to prove that:

\begin{itemize}[leftmargin=1.2em,itemsep=0.15em]
  \item a trusted user intended a semantically correct high-level task;
  \item a trusted source, release validator, or directory returned true information;
  \item the model's prose is accurate, non-toxic, or free from data leakage before a sink;
  \item arbitrary covert channels, denial of service, or compromised operating systems are prevented;
  \item the deterministic benchmark estimates a population-level attack probability.
\end{itemize}

The model is deliberately fail-closed. Missing evidence results in rejection or explicit escalation rather than inferred authority.

\section{Security-Context Discontinuity}
\label{sec:discontinuity}

\begin{definition}[Security-context discontinuity]
Let a security-relevant request traverse representations $E_0,\ldots,E_n$. A discontinuity exists when a fact required to justify the realized effect is (i) absent, (ii) weakened, (iii) reinterpreted under a non-equivalent schema, (iv) modified without an authorized relation, or (v) no longer fresh at the point of effect.
\end{definition}

A discontinuity can arise even when every object is signed. The signature may belong to an unauthorized role; the object may have a valid signature but an invalid relationship to its predecessor; or the permit may bind an approval-time representation that differs from the sink-time action.

\subsection{Minimal counterexample}

Assume a policy gateway approves:
\begin{lstlisting}
operation   = payment.transfer
resource    = account:17
destination = alias:merchant:9
amount      = 5000
source      = verified-invoice:44
\end{lstlisting}
An adapter resolves the alias and emits a signed action. If it outputs \code{bankacct:attacker}, then the final signature authenticates the adapter but does not show that the alias resolution is valid. A name-based tool allowlist also passes because the operation remains \code{payment.transfer}. An action-bound permit issued \emph{after} the adapter likewise binds the attacker's action perfectly. The missing property is an authenticated relation between the approved input field and executed output field.

\system{} requires a transformation witness $w$ and, independently of the adapter, evaluates the consistency predicate:
\[
M_{\textsf{alias}}(v_{\mathrm{in}},v_{\mathrm{out}},w)=\mathsf{true}.
\]
The witness binds the path, both value digests, trusted directory issuer, component, contract, task, and expiry. The directory issuer and the semantics of its signed mapping are part of the TCB: the verifier checks that the attested mapping satisfies the configured relation, not that it can discover the real-world alias mapping on its own. A compromised adapter cannot create a valid directory witness.

\subsection{Fault taxonomy}

The conformance suite instantiates the 32 classes in \cref{tab:faults}. The classes are not claimed to be exhaustive. They are designed to exercise distinct proof obligations and common boundary failures.

\begin{table*}[t]
\centering
\small
\caption{Fault classes in the conformance suite. Each class is instantiated in four domains.}
\label{tab:faults}
\begin{tabularx}{\textwidth}{@{}p{0.16\textwidth}X X@{}}
\toprule
Family & Fault classes & Violated obligation \\
\midrule
Root and grant & Untrusted root producer; root authority exceeded; root scope exceeded; root field-constraint bypass & Accepted ingress/issuer, authenticated root grant, bounded operation/resource/field constraints \\
Provenance and release & Untrusted field binding; provenance drop; provenance value substitution; memory laundering; release predicate bypass; release value substitution; expired release & Signed manifest, leaf-value binding, source continuity, source/value/path/predicate/issuer/expiry-bound release \\
Identity and topology & Principal substitution; unauthorized stage signer; receipt/producer mismatch & Principal continuity, stage--identity--contract binding, authorized predecessor and receipt producer \\
Authority and policy & Authority amplification; delegation widening; taint downgrade; policy downgrade; context-root omission & Non-amplification, scope monotonicity, taint monotonicity, policy freshness, model-context commitment \\
Action semantics & Argument mutation; destination substitution; tool-server swap; effect-class downgrade & Leaf preservation, exact tool/manifest/server binding, exact effect classification \\
Transform contracts & Missing transform witness; invalid transform witness; contract-guarantee violation & Declared relation, trusted witness, before/after binding, postcondition satisfaction \\
Finality state & Subject substitution; post-permit action substitution; revoked grant & Caller binding, exact action binding, finality-time policy/revocation recheck \\
Lifecycle and mediation & Nonce replay; retry duplication; alternate path & One-shot nonce, idempotency, complete effect mediation \\
\bottomrule
\end{tabularx}
\end{table*}

\subsection{Four discontinuity operators}

We use four operators to reason about the taxonomy.

\paragraph{Truncation.}
A projection $q(E)$ omits a required field. If a downstream decision distinguishes two states only through that field, then $q$ makes them observationally equivalent.

\paragraph{Amplification.}
An output contains authority, delegation scope, or declassified data not justified by the input plus an explicit grant. Natural-language derivation is not a grant.

\paragraph{Rebinding.}
An approval, release, or permit is reused for a different principal, task, source value, field, tool, destination, action, audience, or policy epoch.

\paragraph{Staleness or replay.}
An otherwise valid object is used after expiry, revocation, policy change, nonce consumption, or an already-committed idempotency key.

\subsection{Why local checks miss the problem}

A gateway sees $E_i$ and a sink sees $E_j$. If no authenticated transition links them, the sink cannot distinguish a valid transformation from a malicious one. Re-signing at every hop does not solve this: a faulty component can sign an arbitrary output. Safe composition therefore needs both endpoint authenticity and a checked transition relation. The transition relation includes structural preservation, semantic transformation predicates, role authorization, and assumption--guarantee compatibility.

\section{Formal Model}
\label{sec:formal}

\subsection{Envelopes, roots, and paths}

A security-context envelope is
\[
E=(K,A,R,s,q,\sigma),
\]
where $K$ is the context, $A$ the structured action, $R$ a protocol-specific representation, $s$ the producer identity, $q$ a sequence number, and $\sigma$ a signature over the canonical unsigned envelope. Security fields are addressed by canonical paths $f\in\fieldset$; the implementation uses JSON Pointer paths such as \code{/action/parameters/amount\_cents}. Let $E[f]$ denote path resolution and $\Delta(E,E')$ the set of leaf security paths whose values differ.

A signed root grant is
\[
\Gamma=(p,a,t,U,S,T,V,B,\Pi,P,C,n,x,\sigma_\Gamma),
\]
where $p$ is the principal, $a$ the agent actor, $t$ the task root, $U$ an authority set, $S$ a delegation scope, $T$ allowed tool/server/effect identities, $V$ a maximum data classification, $B$ bounded field predicates, $\Pi$ the policy identity/digest/epoch, $P$ and $C$ provenance and context commitments, $n$ a nonce, and $x$ expiry. A deployment accepts only configured ingress identities and grant issuers.

\begin{definition}[Authenticated origin]
$\mathsf{RootOK}(E_0,\Gamma)$ holds iff the envelope and grant signatures verify; their producers have the configured ingress and grant-issuer roles; identity, task, grant identifier, policy, provenance, context, and nonce bindings agree; $E_0$'s authority and scope are subsets of $\Gamma$; its tool, server, effect class, and data class are permitted; every root field predicate holds; and neither object is expired or revoked.
\end{definition}

This definition prevents the circular argument that ``authority never increased relative to an arbitrary $E_0$.'' Monotonicity begins only after a separately authenticated root grant.

\subsection{Provenance and releases}

A provenance manifest contains source records $(id,c,h)$ and claims $(f,id,h_v)$, binding path $f$ to a source identity and the digest of the value at $E_0[f]$. Its signed digest is $P$. A context manifest commits to the source-classified items visible to the planning context; its signed digest is $C$. The verifier recomputes both manifest digests, verifies issuer roles and signatures, resolves every claimed path, and checks each leaf-value digest.

Untrusted sources have classes \textsc{External}, \textsc{ToolResult}, or \textsc{Memory}. Such a value may influence a protected action field only through a release credential
\[
L=(p,a,t,P,id_s,h_s,f,h_v,\phi,op,tool,n,x,\sigma_L),
\]
which binds the principal, actor, task, provenance manifest, source identity and digest, target path, exact value digest, bounded predicate $\phi$, action operation, tool, nonce, expiry, and trusted release issuer.

\begin{definition}[Valid typed release]
$\mathsf{ReleaseOK}(L,E_0,M)$ holds iff $L$ has a trusted, valid, non-revoked signature; all identity, task, manifest, source, path, operation, and tool bindings match; $H(E_0[f])=h_v$; the source claim in $M$ has the same value digest; $\phi(E_0[f])$ is true; and $L$ is unexpired.
\end{definition}

A release is therefore not a field-name whitelist. It authorizes one validated source value to influence one named target field under one bounded predicate and task context.

\subsection{Component contracts}

A component contract is
\[
\contract_i=(F_i^{req},Q_i^{req},A_i,G_i,F_i^{ens},Q_i^{ens},P_i,M_i),
\]
where $F^{req}$ and $F^{ens}$ are required and ensured fields; $Q^{req}$ and $Q^{ens}$ are deterministic predicates on input and output envelopes; $A_i$ are required upstream guarantee tags; $G_i$ are guarantee tags produced on success; $P_i$ is a set of path roots that must be preserved; and $M_i$ maps transformable paths to independently evaluated relations.

A deployment also fixes a pipeline binding
\[
B_i=(stage_i,role_i,signer_i,contract_i).
\]
Thus, a key valid for one role is not accepted for another stage or contract.

A transition receipt is
\[
\begin{aligned}
\rho_i=(&stage,role,id_\contract,H(\contract_i),H(E_i),H(E_{i+1}),\Delta_i,\\
         &W_i,Q_i^{req},Q_i^{ens},signer,\sigma_i).
\end{aligned}
\]
For a changed path $f$, a transformation witness $w\in W_i$ binds the relation identifier, $f$, both value digests, principal, task, component signer, contract, parameters, expiry, and an independently trusted validator signature.

\begin{definition}[Valid transition]
$\mathsf{StepOK}(E_i,\rho_i,E_{i+1},\contract_i,B_i)$ holds iff:
\begin{enumerate}[leftmargin=1.4em,itemsep=0.15em]
  \item the envelope and receipt signatures, input/output digests, sequence, and recomputed change set agree;
  \item the receipt signer equals the output producer and matches $B_i$'s authorized stage, role, signer, contract, and predecessor;
  \item all required fields, required guarantee tags, and input predicates hold;
  \item for every changed security path $f$, either $f$ is not in a preserved root and a declared relation $M_i(f)$ holds with a valid witness, or the step is rejected;
  \item authority and delegation do not widen, taint does not decrease, and policy epoch does not decrease except through a separately authorized policy transition;
  \item all ensured fields and output predicates hold, after which $G_i$ becomes available downstream.
\end{enumerate}
\end{definition}

In the prototype, policy changes are not a permitted component transformation; any changed policy identity, digest, or epoch must still equal current deployment state at admission and finality.

\subsection{Permits and finality}

After verifying the bundle, the verifier issues
\[
\begin{aligned}
\pi=(&p,a,t,\Gamma,sink,H(A_n),H(bundle),\Pi,\\
      &n,k,x,one,\sigma_\pi),
\end{aligned}
\]
where $k$ is an idempotency key. A finality sink validates the permit issuer, caller subject, audience, exact action digest, current policy identity/digest/epoch, current revocation state, expiry, nonce, and idempotency state before committing an effect.

\begin{definition}[Effect witness]
An effect witness is
\[
\begin{aligned}
W_e=(&\Gamma,M_P,M_C,L^*,E_0,\rho_0,E_1,\ldots,\\
     &\rho_{n-1},E_n,\pi,\omega),
\end{aligned}
\]
including the root grant, provenance and context manifests, any releases and transformation witnesses, the transition chain, permit, and outcome receipt.
\end{definition}

\begin{definition}[End-to-end consequence integrity]
A system satisfies \eci{} iff for every externally realized effect $e$, there exists $W_e$ such that
\[
\verify(W_e)=1
\quad\land\quad
\realize(\canon(W_e))=e,
\]
and the witness's authorization was current and unconsumed immediately before $e$ was committed.
\end{definition}

\subsection{Composition conditions}

For a pipeline of $n$ components, the sufficient conditions are:

\begin{description}[leftmargin=1.2em,style=nextline,itemsep=0.2em]
  \item[C1 Authenticated origin.] $\mathsf{RootOK}(E_0,\Gamma)$ and authenticated provenance/context commitments hold.
  \item[C2 Contract compatibility.] Each required guarantee of $\contract_{i+1}$ is established by the root or a valid predecessor, and all deterministic assumptions and postconditions are checked.
  \item[C3 Context continuity.] Every security-critical changed path is preserved or justified by a declared, independently verified relation witness.
  \item[C4 Non-amplification.] Authority and delegation can increase only through a separately authenticated grant, data taint does not decrease across ordinary transitions, and any protected field sourced from untrusted data requires a valid typed release bound to the source value and target field.
  \item[C5 Canonical action binding.] The permit and sink bind the same canonical action, including tool/server manifest, operation, resource, destination, parameters, effect class, and data class.
  \item[C6 Fresh finality.] The sink rechecks subject, audience, policy, revocation, expiry, nonce, and idempotency immediately before effect.
  \item[C7 Complete mediation.] Every path capable of realizing an effect in the protected equivalence classes traverses a compatible finality sink.
\end{description}

\begin{theorem}[End-to-end composition safety]
\label{thm:composition}
Assume unforgeability of the signature scheme and collision resistance of the digest. If C1--C7 hold for a task, then every realized protected effect has a valid effect witness and the system satisfies \eci{}, even if the planner and all attacker-writable content are adversarial.
\end{theorem}

The proof proceeds by induction over valid transition receipts, followed by finality binding; \cref{app:proofs} gives the full argument.

\begin{theorem}[Context truncation]
\label{thm:truncation}
Let a downstream decision $D$ depend on security field $f$. If an intermediate projection $q$ omits $f$, then there exist envelopes $E,E'$ with $q(E)=q(E')$ but $D(E)\ne D(E')$. A downstream component lacking authenticated recovery of $f$ must either fail closed or risk accepting an invalid state.
\end{theorem}

\begin{theorem}[Unbound-field substitution]
\label{thm:substitution}
If a security-critical field $f$ affects the realized effect but is not jointly bound by admission and finality, then an adversary controlling the intervening representation can substitute $A[f]$ with $A'[f]\ne A[f]$ while preserving all checked permit fields.
\end{theorem}

\begin{theorem}[Authority non-laundering]
\label{thm:laundering}
If every ordinary derivation step preserves or reduces authority and only a verified grant edge may add authority, then no finite sequence of LLM generation, summarization, memory consolidation, or protocol translation can create execution authority absent from the authenticated root and grants.
\end{theorem}
In \system{}, this premise is enforced by construction of $\mathsf{StepOK}$: an ordinary transition is accepted only when authority is preserved or reduced, while an increase requires a separately verified grant edge.

\begin{proposition}[Incomplete mediation]
\label{prop:mediation}
If a protected effect class contains any reachable path that bypasses a compatible finality sink, then the system does not satisfy \eci{} for that class, regardless of the correctness of the mediated path.
\end{proposition}

\section{Design}
\label{sec:design}

\begin{figure*}[t]
\centering
\resizebox{0.98\textwidth}{!}{%
\begin{tikzpicture}[
  node distance=0.55cm and 0.6cm,
  box/.style={draw,rounded corners,minimum height=0.75cm,align=center,fill=gray!4},
  trust/.style={draw,rounded corners,minimum height=0.75cm,align=center,fill=blue!6},
  effect/.style={draw,rounded corners,minimum height=0.75cm,align=center,fill=orange!8},
  arrow/.style={-{Latex[length=2mm]},thick}
]
\node[trust] (grant) {Root grant\\and deployment policy};
\node[trust,below=of grant] (prov) {Provenance/context\\manifests and releases};
\node[box,right=1.0cm of grant] (ingress) {Trusted ingress\\$E_0$};
\node[box,right=of ingress] (memory) {Memory stage\\$E_1,\rho_0$};
\node[box,right=of memory] (gateway) {Policy stage\\$E_2,\rho_1$};
\node[box,right=of gateway] (adapter) {Protocol adapter\\$E_3,\rho_2,w$};
\node[trust,right=of adapter] (verify) {Continuity\\verifier};
\node[effect,right=of verify] (sink) {Finality sink\\$\pi\rightarrow e,\omega$};
\draw[arrow] (ingress)--(memory);
\draw[arrow] (memory)--(gateway);
\draw[arrow] (gateway)--(adapter);
\draw[arrow] (adapter)--(verify);
\draw[arrow] (verify)--node[above,sloped]{one-shot permit}(sink);
\draw[arrow] (grant.east)--(ingress.west);
\draw[arrow] (grant.east) to[bend left=18] (verify.north west);
\draw[arrow] (prov.east) to[bend right=12] (ingress.south west);
\draw[arrow] (prov.east) to[bend right=20] (verify.south west);
\node[below=0.12cm of adapter,align=center,font=\scriptsize] {Every changed security path is preserved\\or carries a trusted relation witness};
\node[below=0.12cm of sink,align=center,font=\scriptsize] {Recheck subject, action, policy,\\revocation, expiry, nonce, idempotency};
\end{tikzpicture}%
}
\caption{\system{} architecture. The planner is outside the trusted path and may propose arbitrary actions. Security comes from authenticated roots, proof-carrying transitions, and a revalidating finality boundary.}
\label{fig:architecture}
\end{figure*}
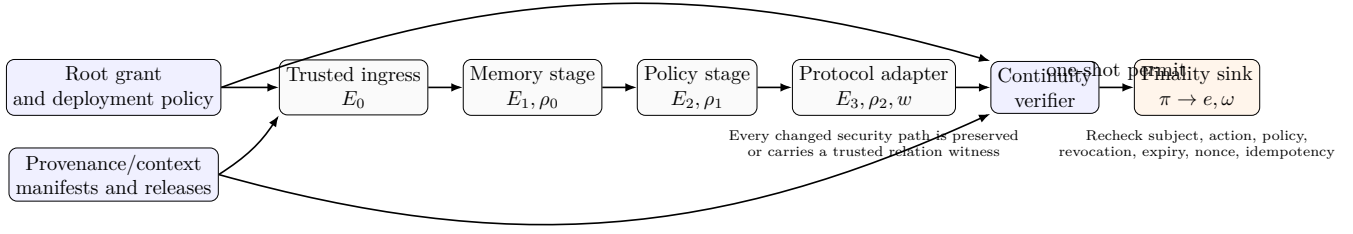

\subsection{Authenticated chain origin}

The deployment policy names accepted ingress identities and root-grant issuers separately. The first envelope must be signed by a trusted ingress, while the root grant must be signed by a trusted authority. The verifier checks that the root's identity, task, authority, delegation, tool/server, effect class, policy, provenance/context commitments, field constraints, nonce, expiry, and revocation state agree with the grant.

This two-role design prevents two attacks from the first artifact. First, a non-ingress component key cannot create a self-signed $E_0$ and then sign the rest of the chain. Second, $E_0$ cannot simply declare a large authority set and rely on downstream monotonicity. The accepted origin is bounded by a separately authenticated grant.

Field constraints are deterministic predicate specifications. For example, a finance task may permit \code{/action/parameters/amount\_cents} only in $[1,10^6]$, while an external-data release for the same field is more restrictive, such as $[1,2\cdot10^4]$. The root grant bounds the overall task; the release bounds the particular untrusted value.

\subsection{Field-resolved provenance}

The provenance manifest contains source records and path claims. A claim binds
\[
(path,source\_id,H(value)).
\]
The verifier parses the path, resolves the actual leaf in $E_0$, and recomputes the digest. Thus, \code{/action/parameters/amount\_cents} is not an opaque label. A manifest that claims a source for one amount but accompanies another amount fails before transition verification.

The context manifest similarly commits to the source-classified items supplied to the planning context. In this prototype, the verifier authenticates the signed manifest and its root; the surrounding runtime is responsible for ensuring the manifest actually covers the bytes shown to the model. This trust boundary is stated explicitly in \cref{sec:limitations}.

\subsection{Bounded typed releases}

Untrusted content is data, not authority. Nevertheless, useful tasks often need external values to determine an action field: an invoice amount, a deployment commit, a recipient alias, or a delegated capability. A blanket rule preventing all such influence would destroy utility.

A typed release solves this by binding one source value to one target path. The credential includes source and manifest digests, exact value digest, principal, actor, task, operation, tool, predicate, issuer, nonce, and expiry. The verifier checks all fields and evaluates the predicate on the value resolved from $E_0$. Examples in the artifact are:

\begin{itemize}[leftmargin=1.2em,itemsep=0.15em]
  \item finance: integer amount in $[1,20{,}000]$ cents;
  \item workspace: destination alias with prefix \code{alias:workspace:customer:};
  \item DevOps: exact signed 40-hex commit identifier;
  \item delegation: capability in the enum \code{\{report.generate\}}.
\end{itemize}

A release signed for value $v$ cannot be replayed for $v'$, even if both satisfy the predicate, because the credential also binds $H(v)$. An expired or revoked credential is rejected.

\subsection{Role-bound transition receipts}

For each stage, the deployment fixes the expected role, signer, and contract. A receipt commits to the input and output envelope digests, contract digest, recomputed changed paths, assumption and guarantee predicate identifiers, and transformation-witness identifiers. Verification requires:

\begin{enumerate}[leftmargin=1.2em,itemsep=0.15em]
  \item receipt signer equals output producer;
  \item signer, role, contract, stage order, and predecessor match deployment;
  \item contract assumptions and required upstream guarantees hold;
  \item all security changes conform to preservation or a declared relation;
  \item contract postconditions hold before its guarantee tags are made available downstream.
\end{enumerate}

A component cannot select an arbitrary trusted contract. The contract is authorized for a specific stage and signer by deployment policy.

\subsection{Semantic transformations}

Real pipelines transform action representations. Alias resolution, unit conversion, schema translation, and tenant routing may legitimately change a security field. Treating every change as malicious makes the system unusable; treating a path listed in \code{may\_transform} as freely mutable makes the declaration meaningless.

\system{} therefore attaches a relation identifier to each transformable path. For a destination alias, the independently signed witness contains the alias, resolved address, both value digests, trusted directory issuer, component signer, contract, principal, task, and expiry. Independently of the adapter, the verifier evaluates:
\[
\begin{aligned}
M_{alias}(v,v',w)\equiv{}& v=w.alias \land v'=w.resolved \\
                          &{}\land\;v\text{ is logical} \land v'\text{ is canonical}.
\end{aligned}
\]
It also verifies that the witness issuer is trusted for \code{alias\_resolution}. This establishes consistency with a directory-attested mapping; the trusted directory and relation registry remain part of the TCB and are not reconstructed by the verifier. Every normal benchmark instance exercises this path. Missing witness, untrusted issuer, false relation, digest mismatch, or postcondition failure rejects the chain.

\subsection{Assume--guarantee flow}

Contracts expose both structural fields and semantic guarantee tags. The root verifier establishes \code{root-authenticated}, \code{provenance-authenticated}, and \code{context-authenticated}. The memory contract requires those tags and, after its predicates and transition checks pass, establishes \code{memory-context-preserved}. The gateway requires that guarantee and establishes \code{policy-authorized}. The adapter requires the gateway guarantee and establishes \code{canonical-action} only if its destination postcondition and transformation relation hold.

A failed stage does not emit usable downstream guarantees. This prevents a later component from satisfying its assumptions merely because a string label appeared in a receipt.

\subsection{Finality and lifecycle state}

The verifier issues a permit only after the whole proof bundle passes. The permit binds the principal, actor, task, root grant, sink audience, exact action digest, proof-bundle digest, current policy identity/digest/epoch, nonce, idempotency key, expiry, and one-time flag.

The finality sink receives the actual caller identity and current runtime state. Immediately before effect it rechecks the signature, caller subject, audience, action digest, policy, revocation, expiry, nonce, and idempotency state. A retry with a new nonce but the same idempotency key returns the prior outcome rather than creating a second effect. A permit issued before a grant is revoked is rejected at finality.

\subsection{Complete mediation}

The proof protects only sinks that require a valid permit. A deployment must enumerate protected effect classes and route every equivalent interface through a compatible broker. The artifact models a direct alternate path as a negative control. Disabling mediation causes the effect even though every receipt and permit mechanism on the normal path remains correct.

\section{Reference Implementation}
\label{sec:implementation}

The artifact is a Python 3.11+ reference implementation. The security core is 1,523 lines in \code{src/continuity/core.py}; scenario construction, fault injection, reference configurations, aggregation, and timing add 1,697 lines in \code{experiment.py}. The repository includes 30 regression tests and publication scripts. It does not invoke an LLM: attack candidates are instantiated directly because the planner is adversarial in the model.

\subsection{Canonical objects and signatures}

All signed dataclasses are serialized to deterministic JSON with sorted keys, integer-only benchmark numerics, normalized sets, and no non-finite floats. SHA-256 supplies object digests, and Ed25519 supplies signatures. This restricted encoder is sufficient for internal reproducibility but is not claimed to be a complete RFC~8785 implementation~\cite{rfc8785}. Production interoperability requires a standardized canonicalizer, explicit Unicode normalization, key rotation, hardware-backed keys, and certificate or workload-identity validation.

The keyring contains public keys, but trust is role-specific in \code{DeploymentPolicy}. Merely appearing in the keyring does not authorize a key to act as ingress, authority, provenance issuer, release validator, transformation validator, component stage, verifier, or sink.

\subsection{JSON Pointer and leaf changes}

Security fields use RFC~6901-style paths~\cite{rfc6901}. Fixed context and action fields are atomic, while nested action parameters are recursively enumerated at leaf granularity. For example, changing only an amount produces:
\begin{lstlisting}
/action/parameters/amount_cents
\end{lstlisting}
rather than a change to the undifferentiated parameters object. The same resolver is used for provenance claims, root constraints, release targets, transform rules, and change-set recomputation. Missing paths fail closed.

\subsection{Verifier algorithm}

\begin{figure}[t]
\begin{lstlisting}
Verify(bundle, deployment, now):
  check root-envelope and root-grant signatures
  enforce trusted ingress and grant issuer
  bind identity, task, authority, scope, policy,
       provenance/context roots and field bounds

  verify provenance and context manifests
  resolve each claimed leaf and compare value digest
  verify releases and evaluate bounded predicates

  guarantees <- authenticated-root guarantees
  for each deployed stage i:
      check signer/role/contract/order/predecessor
      verify envelope and transition receipt
      recompute changed leaf paths
      check required fields, assumptions, guarantees
      for each changed security path:
          require declared transform relation
          verify trusted witness and relation predicate
      check monotonic authority/scope/taint/policy
      check ensured fields and postconditions
      publish stage guarantees only on success

  check final authority, scope, tool manifest,
       current policy, expiry and destination
  return decision
\end{lstlisting}
\caption{Simplified verification procedure. Every check is deterministic; model-generated explanations are ignored.}
\label{fig:verifier}
\end{figure}

The actual verifier accumulates reason codes rather than stopping at the first error. This improves auditability and supports targeted regression tests. Duplicate identifiers, chain-length mismatch, sequence gaps, false change sets, unused or unknown transform witnesses, and undisclosed assumption/guarantee identifiers are also rejected.

\subsection{Finality implementation}

The permit issuer receives only an allowed decision and the final envelope. The sink independently verifies the permit signature and runtime state. The caller subject is an explicit input to \code{FinalitySink.execute}; it is not inferred from permit contents. The sink stores consumed nonces and outcome receipts in an in-memory ledger. On an idempotent retry it returns the existing signed outcome without applying the action again. Effects are recorded only in an in-memory simulated world.

\subsection{Reference pipeline}

The evaluated pipeline has three role-bound stages:

\begin{description}[leftmargin=1.2em,style=nextline,itemsep=0.15em]
  \item[Memory.] Preserves identity, provenance, root, action, and context fields; establishes \code{memory-context-preserved}.
  \item[Gateway.] Requires that guarantee, checks current policy as a deterministic postcondition, and establishes \code{policy-authorized}.
  \item[Adapter.] Requires gateway authorization and resolves the logical destination to a canonical address under an \code{alias\_resolution} relation witness; establishes \code{canonical-action}.
\end{description}

Authority, delegation scope, taint, and policy fields have independent global monotonicity checks in addition to contract logic. This intentional redundancy explains why disabling one monotonicity check does not always create an attack when another invariant still rejects the same mutation.

\subsection{Static contract linting}

A lightweight linter checks field availability and guarantee-tag flow before runtime. It ensures that a stage's required fields and guarantees are available from the authenticated root or predecessors. Runtime verification remains authoritative because static tags alone cannot validate concrete signatures, values, predicates, or transformation relations.

\subsection{Artifact organization}

The artifact includes raw per-run CSV files, per-fault summaries, ablations, performance measurements, figures, tests, and a claim-to-evidence guide. The public artifact is available at \url{\artifacturl}. The code is distributed under the \codelicense{} License; copyright holders are \codecopyrightholders{}.

\section{Evaluation}
\label{sec:evaluation}

We ask five research questions:

\begin{description}[leftmargin=1.2em,style=nextline,itemsep=0.15em]
  \item[RQ1.] Do individually plausible but incomplete control compositions prevent harmful effects?
  \item[RQ2.] Does the full system contain each modeled discontinuity class across domains?
  \item[RQ3.] Can bounded external data and legitimate semantic transformations retain utility?
  \item[RQ4.] Which invariants are exposed by targeted ablations?
  \item[RQ5.] What latency and proof-size overhead does the reference implementation add?
\end{description}

\subsection{Methodology}

\paragraph{Domains.}
The deterministic generator covers workspace email, finance payments, DevOps deployment, and multi-agent delegation. Each action contains domain-specific operations, resources, destinations, parameters, effect classes, and data classifications.

\paragraph{Faults and scale.}
The suite instantiates the 32 fault classes in \cref{tab:faults} in all four domains, with 20 parameterized instances for each fault--domain pair:
\[
32\times4\times20=2{,}560\text{ attack instances}.
\]
It also contains 400 direct benign tasks, 300 benign tasks using signed typed releases, and 200 ambiguous tasks containing an unreleased external field. The total is 3,460 scenarios. Seven system configurations produce 24,220 system--scenario runs.

The 20 instances vary identifiers, principals, actors, tasks, resources, parameter values, nonces, digests, and signatures, but share a deterministic fault schema. They are not independent samples from a real-world attack distribution. We therefore report exact conformance counts and rates over the generated suite, not binomial confidence intervals or population-level security claims.

\paragraph{Adversarial planner.}
Every attack begins after compromise of planning: the candidate action or security transition is directly instantiated in the attacker's desired form. This avoids crediting the security boundary for an LLM that happens to refuse a payload.

\paragraph{Reference configurations.}
The six incomplete configurations are mechanism-level negative controls, not reimplementations of named systems:

\begin{itemize}[leftmargin=1.2em,itemsep=0.15em]
  \item \emph{PassThrough}: directly executes the final action;
  \item \emph{ToolAllowlist}: checks only the final tool identifier;
  \item \emph{GatewayPolicy}: evaluates authority, scope, policy, and taint at the gateway, then trusts downstream changes;
  \item \emph{ProvenanceGateway}: additionally validates field provenance and releases at the gateway;
  \item \emph{EffectBoundPermit}: binds the post-adapter action to finality without validating its upstream justification;
  \item \emph{Gateway+Finality}: combines gateway policy with an action-bound finality permit, but does not authenticate the gateway-to-adapter transition.\footnote{The permit is bound to the final post-adapter canonical action because benign traces include an authorized alias-to-canonical transformation. Binding the permit to the pre-transform gateway representation would reject valid traces rather than isolate the remaining composition gap.}
\end{itemize}

\paragraph{Metrics.}
\emph{Effect attack success rate} (effect ASR) is the fraction of attack instances that commit a harmful effect. For replay and retry faults, failure requires more than one committed effect; a single authorized effect followed by a rejected or idempotent retry is correct. \emph{Benign auto-completion} is the fraction of benign tasks that commit one effect without escalation. \emph{Ambiguous escalation accuracy} is the fraction of unreleased external-field tasks that commit no effect and return \textsc{Escalate}. \emph{Lifecycle correctness} compares both expected terminal state and expected effect count, rather than treating every attack as requiring the literal disposition \textsc{Deny}.

\subsection{Main results}

\begin{table}[t]
\centering
\scriptsize
\setlength{\tabcolsep}{3pt}
\caption{Main conformance results. ``Classes'' is the number of contained fault--domain pairs out of 128.}
\label{tab:main-results}
\begin{tabular}{@{}lrrrr@{}}
\toprule
Configuration & Effect ASR & Classes & Benign & Ambig. \\
\midrule
Pass-through          & 100.0\% &   0 & 100\% &   0\% \\
Tool allowlist        & 100.0\% &   0 & 100\% &   0\% \\
Gateway policy        &  96.9\% &   4 & 100\% &   0\% \\
Provenance gateway    &  84.4\% &  20 & 100\% & 100\% \\
Effect-bound permit   &  81.3\% &  24 & 100\% &   0\% \\
Gateway + finality    &  65.6\% &  44 & 100\% & 100\% \\
\textbf{CONTINUITY}   & \textbf{0.0\%} & \textbf{128} & \textbf{100\%} & \textbf{100\%} \\
\bottomrule
\end{tabular}
\end{table}

\begin{figure*}[t]
\centering
\includegraphics[width=0.90\textwidth]{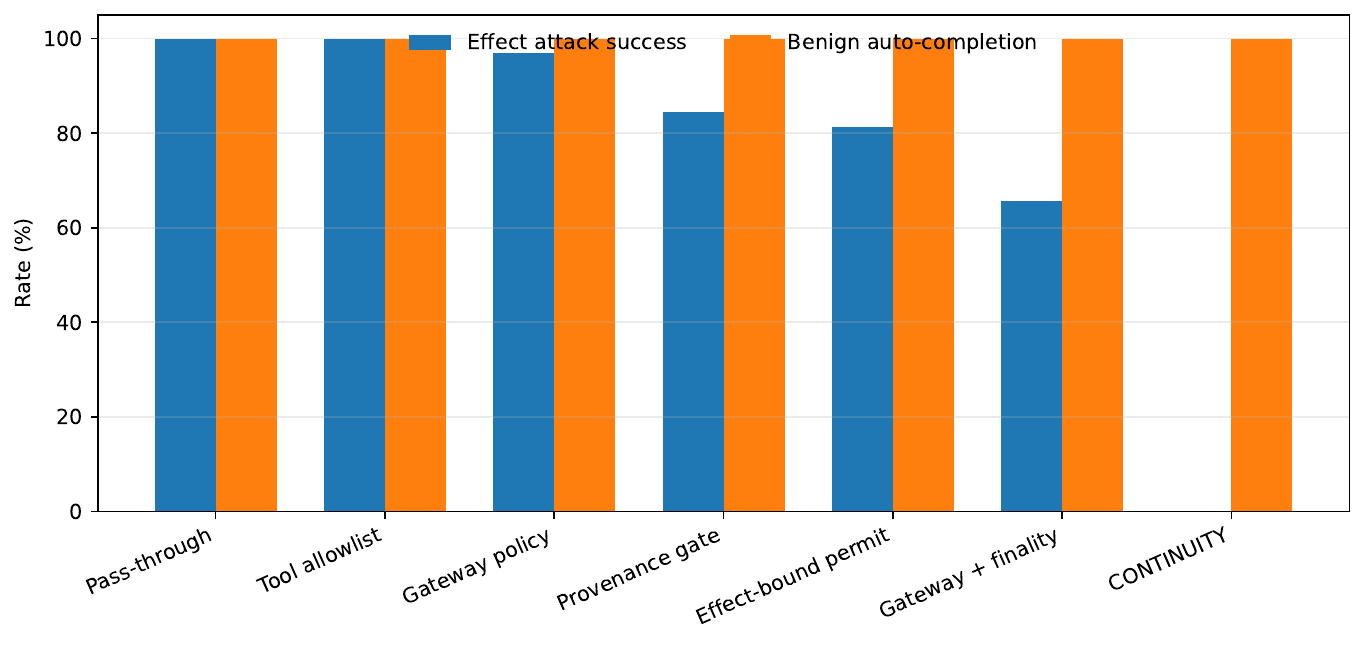}
\caption{Harmful-effect success and benign auto-completion. The full system preserves benign completion while incomplete compositions remain vulnerable to cross-layer faults.}
\label{fig:main-results}
\end{figure*}

\system{} commits no harmful effect in 2,560 attack instances and contains every one of the 128 fault--domain classes. It completes all 700 benign tasks and escalates all 200 ambiguous tasks. Its lifecycle correctness is 100\%, including replay and retry semantics.

The strongest incomplete configuration, Gateway+Finality, still has 65.6\% effect ASR. This result illustrates the central composition claim. Gateway authorization plus exact final action binding is insufficient when the permit issuer accepts a post-adapter action without verifying how it relates to the gateway-approved state. Conversely, EffectBoundPermit protects replay, subject, revocation, and post-permit substitution but lacks root, provenance, and contract justification; its ASR is 81.3\%.

GatewayPolicy contains only four fault--domain classes because it observes the pre-adapter state. ProvenanceGateway contains 20 classes and correctly escalates unreleased external fields, but downstream provenance loss, parameter changes, role misuse, and alternate paths remain outside its view.

\subsection{Utility and validated change}

All normal traces perform a security-relevant adapter transformation. The ingress and gateway carry a logical destination alias; the final action carries a canonical address. The full system accepts the transformation only with a trusted directory witness and verified alias relation. Thus, 100\% benign completion does not result from a preserve-everything contract.

All 300 external-data tasks carry a signed typed release. The verifier checks source identity and digest, leaf path, exact value digest, operation, tool, bounded predicate, issuer, task, expiry, and revocation. All 300 complete. In contrast, all 200 otherwise plausible tasks that omit a release return \textsc{Escalate} and commit no effect. The release attack classes demonstrate negative controls: an out-of-bound value, value substitution, or expired credential is rejected even when the target field name is present.

\subsection{Ablation study}

The ablation suite uses one representative instance for each of the 128 fault--domain pairs. \cref{tab:ablation} reports reopened classes. ``NoContractConformance'' retains envelope and receipt signatures, digest links, recomputed change sets, global monotonicity, and finality; it disables preservation, transform-rule, and contract postcondition checks. It is therefore not mislabeled as removal of receipts.

\begin{table}[t]
\centering
\scriptsize
\caption{Targeted ablations over 128 fault--domain classes. Zero for one check can indicate redundant containment by another invariant, not irrelevance to the sufficient-condition theorem.}
\label{tab:ablation}
\begin{tabular}{@{}lr@{}}
\toprule
Ablation & Reopened classes \\
\midrule
No field provenance                  & 24 \\
No contract conformance              & 24 \\
Incomplete mediation                 & 24 \\
No root authentication               & 16 \\
No release validation                & 12 \\
No transform-witness validation      &  8 \\
No replay protection                 &  8 \\
No component-role binding            &  4 \\
No identity binding                  &  4 \\
No delegation monotonicity           &  4 \\
No taint monotonicity                &  4 \\
No policy freshness                  &  4 \\
No context commitment                &  4 \\
No action binding                    &  4 \\
No subject binding                   &  4 \\
No revocation recheck                &  4 \\
No authority monotonicity alone      &  0 \\
\bottomrule
\end{tabular}
\end{table}

The three largest effects are field-provenance removal, contract-conformance removal, and incomplete mediation, each reopening 24 classes. Root authentication reopens 16, including the exact self-origin and over-broad-origin counterexamples. Release validation reopens 12 classes; transform validation and replay protection each reopen eight. Authority monotonicity alone reopens none in this benchmark because operation and action changes remain rejected by root, preservation, scope, and exact-action controls. This demonstrates defense in depth and also cautions against interpreting an ablation as proof that a theorem condition is logically unnecessary.

\begin{figure}[t]
\centering
\includegraphics[width=\columnwidth]{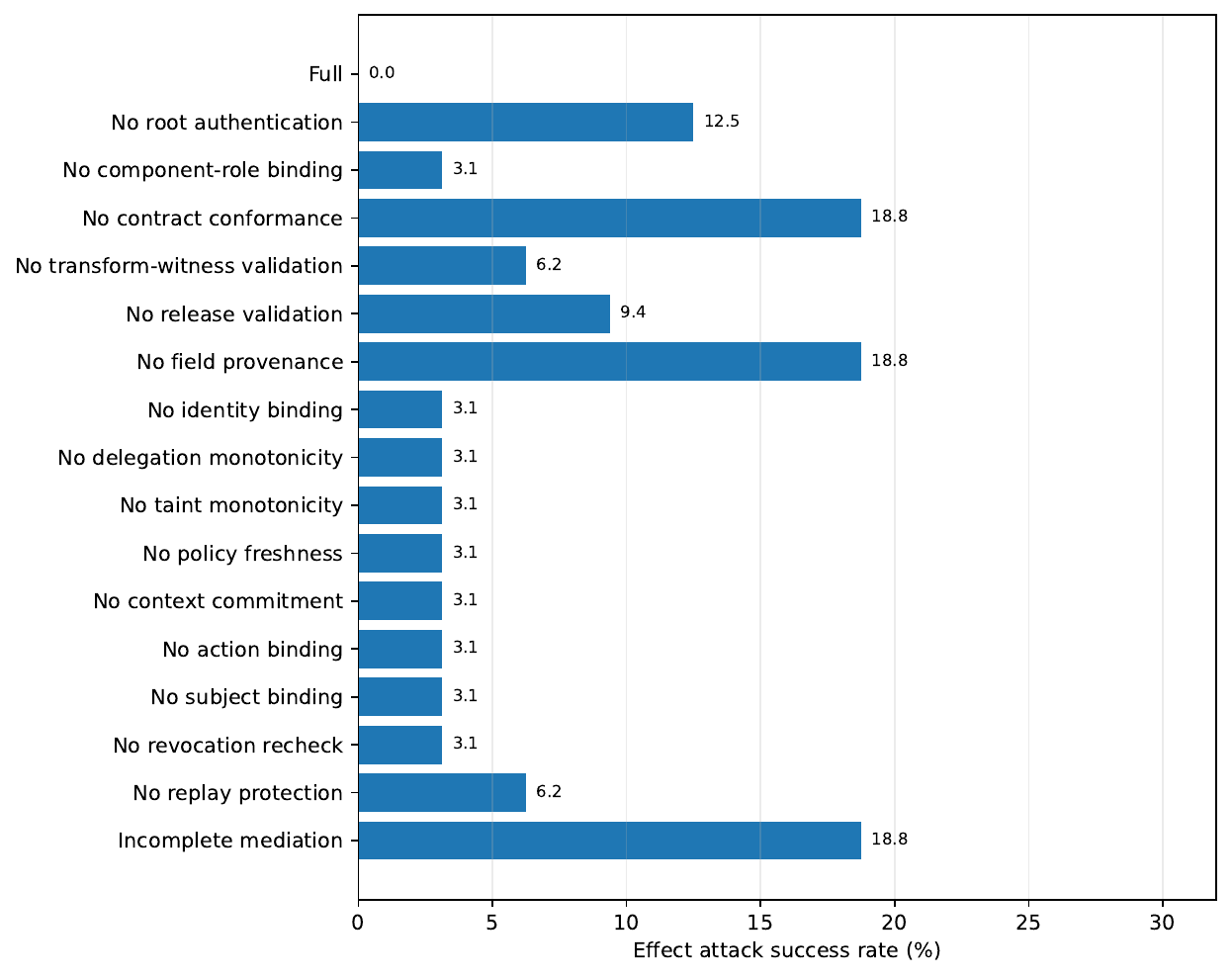}
\caption{Effect ASR under nonzero targeted ablations. Full CSV results retain the zero-valued redundant ablations.}
\label{fig:ablation}
\end{figure}

\subsection{Performance}

Measurements were collected on the recorded Linux host using Python and the \code{cryptography} Ed25519 implementation. Each reported microbenchmark operation is repeated 300 times. Timing is host- and runtime-dependent and excludes network, model, external-policy-service, and durable-storage latency.

\begin{table}[t]
\centering
\small
\caption{Reference-prototype latency.}
\label{tab:performance}
\begin{tabular}{@{}lrr@{}}
\toprule
Operation & p50 & p95 \\
\midrule
Proof verification & 4.21 ms & 4.91 ms \\
End-to-end transition + finality & 7.17 ms & 8.07 ms \\
\bottomrule
\end{tabular}
\end{table}

Verification scales approximately linearly with the number of signed transitions. In the recorded run, bundles with 1, 3, 5, 10, and 20 transitions occupy approximately 8.1, 12.4, 16.8, 27.6, and 49.4~KiB, respectively. The implementation stores full signed envelope snapshots. A production encoding could use deltas, Merkle commitments, checkpoint receipts, batch verification, and compact binary serialization.

\begin{figure}[t]
\centering
\includegraphics[width=\columnwidth]{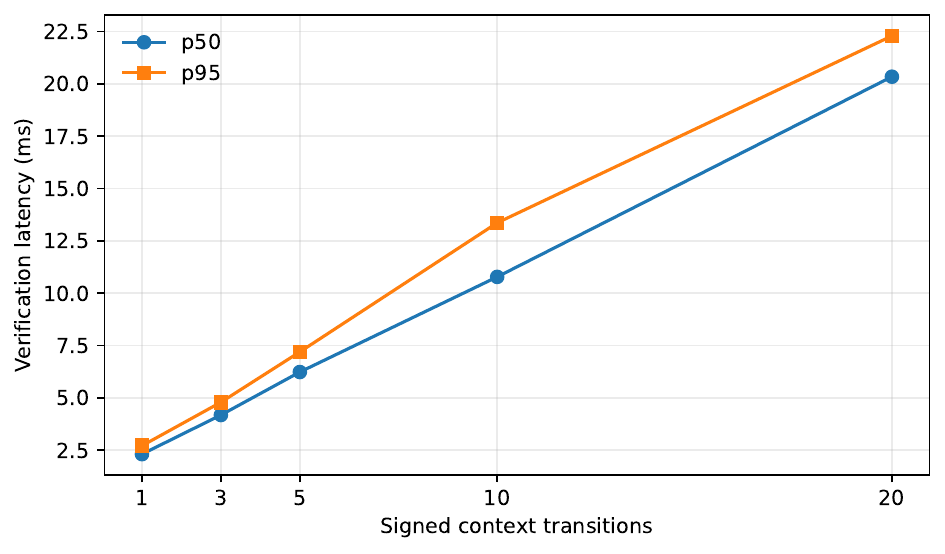}
\caption{Proof-verification latency versus signed transitions.}
\label{fig:scaling}
\end{figure}

\subsection{Interpretation}

The evaluation supports a conformance statement: for every generated instance, the full implementation enforces the modeled obligations, and removing particular obligations exposes corresponding faults. It does not establish a universal attack rate, completeness of the fault taxonomy, correctness of trusted validators, or resistance to a compromised TCB. Those limitations are addressed in \cref{sec:limitations}.

\section{Security Analysis}
\label{sec:security}

\subsection{Theorem-to-implementation correspondence}

\Cref{thm:composition} is conditional. The artifact now enforces each stated condition rather than treating it as an external assumption without a check:

\begin{itemize}[leftmargin=1.2em,itemsep=0.15em]
  \item C1 maps to trusted ingress, signed root grant, role-specific issuer sets, root field constraints, signed provenance/context manifests, and leaf-value digest checks.
  \item C2 maps to pipeline stage bindings, required/ensured guarantee tags, and deterministic input/output predicates.
  \item C3 maps to JSON Pointer leaf changes, preserved roots, declared transform rules, and signed relation witnesses.
  \item C4 maps to root subsets, transition monotonicity, and source/value/path/predicate-bound releases.
  \item C5 maps to a complete structured action digest shared by verification and finality.
  \item C6 maps to finality-time subject, policy, revocation, expiry, nonce, and idempotency checks.
  \item C7 maps to the permit-requiring effect broker; the alternate-path fault demonstrates the consequence of violating it.
\end{itemize}

The implementation does not mechanically prove that its Python code refines the formal model. The regression suite supplies executable evidence for the specified checks and counterexamples; a proof assistant or verified implementation remains future work.

\subsection{Root compromise}

Monotonicity cannot secure an attacker-chosen origin. If a key accepted as a root-grant issuer is compromised, the attacker can issue an over-broad root grant; if both deployment trust and revocation fail to remove that key, the guarantee is lost. Similarly, compromise of a trusted provenance, release, or transformation validator allows false claims within that validator's role.

This is different from compromise of an ordinary pipeline component. A memory, gateway, or adapter key is authorized only at its configured stage and contract. The attacker can sign that stage's output, but cannot start a chain, issue a grant, impersonate a different stage, create a trusted release, or create a trusted transformation witness. It may still exploit any overly permissive relation or predicate configured for its contract. Thus, stage-key compromise is contained only to the extent that independently checked contracts and validators are restrictive.

\subsection{Signatures versus authorization}

The verifier distinguishes key presence from role authorization. Three checks are separate:
\[
\begin{aligned}
&\mathsf{SigValid}(k,o),\qquad \mathsf{RoleAllowed}(k,r),\\
&\mathsf{TransitionOK}(o_{in},o_{out},\contract).
\end{aligned}
\]
An object is accepted only when all required predicates hold. This prevents the attack in which a valid adapter key signs $E_0$ and every receipt, or a receipt is signed by a different valid component than the output producer.

\subsection{Typed-release safety}

A release predicate limits the allowed value set, but predicate satisfaction alone is insufficient. Without an exact value digest, a credential could be replayed for another value that also satisfies the predicate. Without source and manifest binding, the same value could be attributed to a different document. Without principal, task, operation, and tool binding, it could be replayed across workflows. The implemented release checks all of these dimensions.

Trusted-release issuance remains security-critical. \system{} verifies that an issued credential is internally and contextually valid; it does not prove that the validator's extraction from an invoice or other source was semantically correct.

\subsection{Transformation safety}

A contract declaration that a path ``may transform'' is not itself authorization. The relation witness must be signed by an issuer trusted for that relation and bind both before and after values. The verifier independently evaluates the relation predicate. The current artifact implements alias resolution. A production system would require relation-specific validators for currency conversion, schema normalization, tenant routing, or identifier resolution, together with versioned semantics and revocation.

If the relation predicate is non-injective, omits a security field, or admits attacker-chosen outputs, \cref{thm:substitution} applies. The relation registry is therefore part of the TCB.

\subsection{Freshness and time-of-check/time-of-use}

Admission-time correctness is not finality-time correctness. A policy can change, a grant can be revoked, or the caller can be substituted after a permit is issued. The sink rechecks the current state rather than relying solely on the verifier's earlier decision. It also atomically consumes nonce/idempotency state before effect in the model.

A real distributed provider may have its own concurrency, partial-failure, and idempotency semantics. The prototype's in-memory atomicity is a simplifying assumption. Production adapters need provider-specific transactional or compensating behavior.

\subsection{Complete mediation}

The system cannot protect a path it does not mediate. Tool aliases, direct SDK calls, browser automation, shell commands, remote agents, and recovery paths must be mapped to effect classes and routed through compatible sinks. Deployment review should therefore test effect reachability, not only registered tool names. The artifact's \code{ALTERNATE\_PATH} fault commits an effect immediately when mediation is disabled, even though all normal-path cryptographic checks remain intact.

\subsection{Residual attacks}

An adversarial planner can still cause denial of service by repeatedly proposing rejected actions, consume policy and signature-verification resources, or produce unsafe textual output not connected to a protected sink. A malicious but trusted source may supply incorrect facts that satisfy all structural checks. Covert exfiltration through allowed text, timing, resource names, or aggregate queries is not eliminated. \system{} is a consequence-integrity boundary, not a complete information-flow or semantic-safety solution.

\section{Discussion}
\label{sec:discussion}

\subsection{What should be standardized?}

The artifact suggests that interoperability requires more than a universal \code{ALLOW}/\code{DENY} hook. A cross-control contract needs stable semantics for at least:

\begin{itemize}[leftmargin=1.2em,itemsep=0.15em]
  \item principal, actor, task, root grant, and delegation scope;
  \item field-resolved provenance and model-context commitments;
  \item canonical action identities and protected effect classes;
  \item typed releases and transformation-relation witnesses;
  \item policy identity/digest/epoch, audience, nonce, idempotency, and outcome.
\end{itemize}

MCP and runtime-hook standards can carry such objects, but the paper does not propose replacing their transport or authorization layers. The likely integration point is an extension or sidecar envelope associated with tool calls and agent handoffs. An implementation must specify canonicalization and field paths identically at the producer and verifier.

\subsection{Fail closed versus escalation}

Missing proof does not always mean the user requested a forbidden task. An unreleased invoice value may be legitimate but insufficiently validated. The system therefore distinguishes \textsc{Escalate} from \textsc{Deny}. Escalation can request a signed source validator, narrower user confirmation, reauthentication, or human approval bound to the exact action digest. Invalid signatures, role violations, root excess, stale state, and false relation witnesses remain hard failures.

\subsection{Trust distribution}

The prototype stores deployment trust locally. In a single enterprise, ordinary PKI, workload identity, or a transparency log may be sufficient. Cross-organization agent interactions may require shared issuer discovery, revocation, and evidence anchoring. A consortium or public ledger could anchor policy, manifest, revocation, or batch-evidence roots, but it is optional. The safety property comes from authenticated contracts and finality enforcement, not from placing prompts or private task data on a blockchain.

\subsection{Policy and relation evolution}

Contracts, predicates, and relation semantics are versioned security code. Updating \code{alias\_resolution} or a release predicate changes the accepted effect set and should trigger review, compatibility testing, and policy-epoch changes. A permit issued under an older policy is rejected once the sink's runtime state advances. Long-lived workflows need explicit migration or reauthorization rather than silent compatibility.

\subsection{Observability and incident response}

A complete witness supports questions that ordinary traces cannot reliably answer: Which root grant authorized the task? Which source supplied a protected field? Which release and predicate admitted it? Which component changed the destination? Which relation witness justified the change? Which policy epoch and sink committed the effect? These objects can be retained off-chain with privacy controls, while only hashes or batch roots are externally anchored.

\subsection{Relationship to information flow}

\system{} tracks selected provenance and taint but is not a full language-level information-flow control system. It focuses on structured effect fields and explicit transitions. CaMeL and Prompt Flow Integrity address broader control/data-flow separation~\cite{debenedetti2025camel,kim2025pfi}; ARGUS tracks causal influence~\cite{weng2026argus}. These mechanisms can serve as upstream producers of authenticated manifests, while \system{} tests whether their output survives to finality. The composition contract does not substitute for their internal analyses.

\section{Related Work}
\label{sec:related}

\paragraph{Prompt-injection evaluation.}
Early work demonstrated indirect prompt injection against application-integrated LLMs~\cite{greshake2023indirect}. InjecAgent and AgentDojo provide systematic environments for tool-integrated and dynamic agent attacks~\cite{zhan2024injecagent,debenedetti2024agentdojo}. Those benchmarks primarily study whether an attacker can influence agent behavior and task outcomes. Our conformance suite begins from an adversarial proposal and instead targets discontinuities among downstream controls. The two evaluation styles are complementary: a natural-language benchmark estimates proposal pressure, while \system{} tests whether a compromised proposal can cross a structural effect boundary.

\paragraph{Control/data isolation and privilege.}
CaMeL separates trusted control from untrusted data and constrains capabilities through program analysis~\cite{debenedetti2025camel}. Prompt Flow Integrity combines execution isolation, secure data handling, and privilege control~\cite{kim2025pfi}. IsolateGPT isolates executable agent components~\cite{wu2024isolategpt}, and Progent provides programmable privilege control for tools~\cite{shi2025progent}. These works motivate least authority and deterministic enforcement. Our focus is the interface theorem: what authenticated fields and guarantees must be preserved when such a control hands the action to a different policy, adapter, or sink.

\paragraph{Policy enforcement.}
Formal policy work translates policy into enforceable runtime constraints and evaluates practical agent systems~\cite{palumbo2026formal}. The OWASP Agent Control Standard defines runtime interception and policy-decision hooks~\cite{owasp2026acs}. \system{} does not compete with a policy language or hook API. It supplies a proof-carrying transition model in which a policy decision is one stage guarantee, and later components must preserve the facts on which that decision depended.

\paragraph{Provenance and execution integrity.}
ARGUS reconstructs causal runtime provenance and audits whether proposed actions are supported by trusted task information~\cite{weng2026argus}. CXI binds field authority, typed releases, exact effects, and invocation authority at a protected sink~\cite{santosgrueiro2026cxi}. Our design adopts the importance of field-level provenance, validated releases, and exact effect binding. The distinct contribution is composition across heterogeneous controls: trusted root enforcement, role-bound component contracts, signed before/after transition receipts, independently checked semantic relations, and a fault-injection methodology for context loss between provenance, policy, adapter, and finality layers.

\paragraph{Protocols and authorization.}
MCP specifies tool discovery and invocation and cautions that untrusted tool metadata is not itself authoritative~\cite{mcp2026tools,mcp2026security}. OAuth security guidance similarly emphasizes audience, least privilege, sender constraints, freshness, and replay resistance~\cite{rfc9700}. \system{} binds these concepts to a concrete action and task witness but leaves transport and user authentication to existing protocols.

\paragraph{Compositional and supply-chain assurance.}
Assume--guarantee reasoning has a long history in compositional specifications~\cite{abadi1993composing}. Software supply-chain frameworks such as in-toto and TUF authenticate authorized transformations and compromise-resilient metadata rather than trusting only the final artifact~\cite{torresarias2019intoto,samuel2010tuf}. \system{} applies a related insight to agent execution: each security-sensitive transformation must be attributable to an authorized role and satisfy an explicit relation before the final action is trusted.

\paragraph{Positioning.}
The paper does not claim novelty for digital signatures, provenance, capabilities, policy enforcement, typed releases, finality, JSON canonicalization, or assume--guarantee reasoning individually. It contributes a concrete end-to-end composition property, an implementation that enforces previously omitted root/release/transform assumptions, and an executable taxonomy of cross-control discontinuities.

\section{Limitations}
\label{sec:limitations}

\paragraph{Trusted roots are deployment inputs.}
The prototype now enforces trusted ingress, root-grant issuers, root bounds, and role-specific validators. It does not provision those roots or decide which organization should be trusted. Compromise or malicious configuration of a root authority, provenance issuer, release validator, transformation validator, verifier, or mandatory finality sink can invalidate the corresponding guarantee. Key rotation, quorum authorization, hardware roots, and certificate-path validation are outside the prototype.

\paragraph{Integrity is not semantic correctness.}
CONTINUITY guarantees continuity of authenticated security facts and authorization decisions, not the truth of those facts. A trusted provenance, release, or transformation validator may itself produce an incorrect structured conclusion. In that case, CONTINUITY can faithfully preserve and enforce an incorrect conclusion. Semantic correctness of trusted validators therefore remains part of the trusted computing base and is orthogonal to end-to-end consequence integrity.

\paragraph{Context-manifest completeness.}
The verifier recomputes and authenticates the context-manifest digest. The surrounding agent runtime must ensure that this manifest covers all bytes and tool outputs shown to the planner. The prototype does not instrument a production model runtime or prove completeness of its context capture. A runtime that silently exposes additional attacker-controlled content violates the assumption.

\paragraph{Restricted transformation language.}
The implementation evaluates one security-relevant relation, alias resolution, and a small deterministic predicate registry. Real deployments need richer, versioned semantics for currency conversion, schema mapping, identifier resolution, aggregation, and declassification. General-purpose relation languages create their own soundness, decidability, and policy-review challenges.

\paragraph{No verified implementation.}
The proofs apply to the abstract model and conditions. The Python artifact is extensively regression-tested but not mechanically verified, constant-time, hardened, or formally shown to refine the model. Its deterministic JSON encoder is intentionally restricted and should be replaced by a standards-conformant canonicalizer in interoperable deployments.

\paragraph{Synthetic conformance benchmark.}
The suite contains 32 designed fault templates, four domains, and deterministic parameter variations. It is useful for falsifying missing invariants and validating the artifact, but it is neither exhaustive nor sampled from a real attack population. The 0/2,560 result must not be interpreted as an estimated universal failure probability. It also does not compare model quality because no LLM is invoked.

\paragraph{Simplified provider semantics.}
The finality world, nonce ledger, and idempotency store are in-memory and atomic. Real providers exhibit concurrency, retries, eventual consistency, non-idempotent side effects, and partial failures. Provider-specific transaction protocols and durable evidence are required for production.

\paragraph{Incomplete information-flow coverage.}
The system checks declared protected fields and effects. It does not eliminate covert channels, infer every semantically equivalent effect path, or prevent leakage through allowed textual output, timing, aggregate queries, resource names, or malicious downstream provider behavior. Complete mediation is an engineering and governance obligation.

\paragraph{Human and semantic error.}
A valid user grant may authorize an ill-advised action. A human can approve a misleading request, and a trusted tool can return false data. \system{} preserves authorization context; it does not guarantee that the authorized goal is wise, legal, or semantically correct.

\paragraph{Artifact integration.}
The reference pipeline models memory, policy, adapter, and finality boundaries but is not a production MCP, A2A, OWASP ACS, cloud-IAM, or blockchain integration. Those bindings require separate protocol profiles and interoperability testing.

\section{Conclusion}
\label{sec:conclusion}

Agent security controls do not compose automatically. A provenance label, policy decision, capability, or action-bound permit can be locally correct yet fail when a neighboring component drops its assumptions, changes the action under a different schema, accepts an unauthorized signer, or commits an effect under stale state. We formalized this problem as security-context discontinuity and defined end-to-end consequence integrity.

\system{} demonstrates a concrete response: authenticate the root grant; bind each stage to an authorized identity and contract; resolve provenance at leaf paths; make releases source-, value-, predicate-, task-, and expiry-bound; require independently checked relation witnesses for legitimate transformations; and revalidate subject, action, policy, revocation, and replay state at finality. In the deterministic artifact, these checks contain all 128 modeled fault--domain classes while preserving benign execution and explicit escalation.

The broader lesson is architectural: planning may remain probabilistic and adversarial, but external consequence should be mediated by a small, deterministic, proof-carrying control plane whose guarantees remain continuous from instruction context to realized effect.

\appendix
\section{Proofs}
\label{app:proofs}

We state the proof obligations at the abstraction level of \cref{sec:formal}. Cryptographic assumptions are signature unforgeability and digest collision resistance; implementation refinement is not claimed.

\begin{lemma}[Root boundedness]
\label{lem:root}
If $\mathsf{RootOK}(E_0,\Gamma)$ holds, then the principal, actor, task, policy, provenance/context commitments, authority, delegation, tool/server/effect identities, data class, and constrained fields in $E_0$ are authorized by an accepted root issuer.
\end{lemma}
\begin{proof}
The signatures bind $E_0$ and $\Gamma$ to their producers. Deployment role checks restrict those producers to trusted ingress and grant issuers. Equality checks bind the identity, task, policy, commitments, and nonce. Subset and membership checks bound authority, delegation, tools, servers, effects, and classification. Every field constraint is evaluated on the resolved value in $E_0$. Expiry and revocation checks establish current validity. Hence no accepted root fact exceeds $\Gamma$ under the stated cryptographic assumptions.
\end{proof}

\begin{lemma}[Transition preservation]
\label{lem:step}
If $\mathsf{StepOK}(E_i,\rho_i,E_{i+1},\contract_i,B_i)$ holds, then every security-critical difference between $E_i$ and $E_{i+1}$ is either an authorized monotone reduction or satisfies a declared transformation relation with a valid witness; all $G_i$ guarantees are true of $E_{i+1}$.
\end{lemma}
\begin{proof}
The receipt and envelope signatures plus input/output digests bind the exact endpoints, authorized signer, role, contract, stage, and sequence. Recomputed $\Delta(E_i,E_{i+1})$ prevents an omitted changed path. For each changed path, contract verification rejects a preserved path or a path lacking a declared relation. A witnessed path is accepted only if the trusted issuer, relation, path, before/after digests, component, contract, principal, task, and expiry agree and the deterministic relation predicate is true. Independent monotonicity checks reject increased authority or delegation, decreased taint, and policy downgrade. Required and ensured predicates are evaluated concretely; $G_i$ is published only if no stage error occurs.
\end{proof}

\begin{lemma}[Inductive context continuity]
\label{lem:induction}
Suppose $\mathsf{RootOK}(E_0,\Gamma)$ and $\mathsf{StepOK}$ holds for all $i\in[0,n-1]$. Then every security field in $E_n$ is connected to an authorized root value by a sequence of preserved equalities, monotone restrictions, or valid transformation relations, and every downstream assumption is discharged by an authenticated root guarantee or a valid predecessor guarantee.
\end{lemma}
\begin{proof}
By induction on $i$. The base case follows from \cref{lem:root}. For the inductive step, \cref{lem:step} preserves the property for unchanged fields, adds only justified relation edges for transformed fields, and forbids unauthorized amplification. Guarantee publication is conditional on successful verification, so the available guarantee set contains only root guarantees or guarantees from valid prior steps. Therefore the property holds at $E_{i+1}$.
\end{proof}

\begin{lemma}[Finality binding]
\label{lem:finality}
If a finality sink commits $e$ under permit $\pi$, then the actual caller, action, audience, policy state, grant state, expiry, nonce, and idempotency state match $\pi$ immediately before the effect, and $e=\realize(\canon(A_n))$.
\end{lemma}
\begin{proof}
The sink accepts only the configured permit signer and verifies its signature. It compares the explicit caller subject, sink audience, and actual action digest to the permit. It compares policy identity, digest, and epoch to current runtime state and rejects revoked grants or permits. It checks expiry and consumes one-time or idempotency state before applying the action. Collision resistance prevents a different canonical action from sharing the checked digest. The outcome receipt binds the committed effect and resulting state.
\end{proof}

\begin{proof}[Proof of \cref{thm:composition}]
By C1 and \cref{lem:root}, $E_0$ has an authenticated, bounded origin and authenticated provenance/context commitments. The release clause of C4 validates every protected field whose bound source is untrusted. By C2--C4 and \cref{lem:induction}, all fields in the final envelope are causally connected to that origin through valid contract steps without unauthorized authority or data amplification. By C5, the permit binds the same canonical action justified by the bundle. By C6 and \cref{lem:finality}, the sink commits only under current, unconsumed state and the actual action, subject, and audience match. C7 excludes an alternate protected path without such a sink. Concatenating the grant, manifests, releases, envelopes, receipts, transformation witnesses, permit, and outcome yields $W_e$ with $\verify(W_e)=1$ and $\realize(\canon(W_e))=e$. Therefore \eci{} holds.
\end{proof}

\begin{proof}[Proof of \cref{thm:truncation}]
Because $D$ depends on $f$, there exist $E,E'$ identical on all fields except $f$ such that $D(E)\ne D(E')$. Since $q$ omits $f$, $q(E)=q(E')$. Any deterministic downstream function of $q(E)$ alone returns the same decision for both and is therefore wrong for at least one. It must fail closed for both or obtain an authenticated recovery of $f$ to distinguish them.
\end{proof}

\begin{proof}[Proof of \cref{thm:substitution}]
Let $f$ influence the effect but be absent from all fields jointly committed by admission and the permit. Choose authorized action $A$ and $A'$ differing only at $f$ such that their effects differ. Because the checked commitments omit $f$, the same permit validates both checked projections. An adversary controlling the intervening representation can present $A'$ at the sink. If the sink does not rebind the complete action, it commits the unauthorized effect.
\end{proof}

\begin{proof}[Proof of \cref{thm:laundering}]
Induct on the number of ordinary derivation steps. The base authority is bounded by the authenticated root. By construction of $\mathsf{StepOK}$, each accepted non-grant step enforces that output authority is a subset of input authority, so it cannot add an element. At a verified grant edge, only the authority explicitly contained in that valid grant may be added. Hence any authority present after finitely many steps belongs to the root or a verified grant, not to natural-language derivation alone.
\end{proof}

\begin{proof}[Proof of \cref{prop:mediation}]
Let $p$ be a reachable path in a protected effect-equivalence class that bypasses a compatible finality sink. The adversary selects $p$ and invokes the effect without producing a witness accepted by that sink. An effect is realized without the required $W_e$, contradicting \eci{}. Correctness of other paths is irrelevant to reachability of $p$.
\end{proof}

\section{Representative Object Schemas}
\label{app:schema}

The listings omit signatures and some identifiers for space. The implementation signs the full canonical unsigned dataclass.

\subsection{Root grant}
\begin{lstlisting}
{
  "grant_id": "grant:finance:17",
  "issuer": "key:authority",
  "principal": "did:example:user:17",
  "actor": "spiffe://example/agent/finance/4",
  "task_root": "task:finance:17",
  "authority": ["payment.transfer"],
  "delegation_scope": ["payment.transfer:account:17"],
  "allowed_tool_ids": ["urn:mcp:finance:payments"],
  "allowed_server_ids": ["spiffe://finance/payments"],
  "field_constraints": [{
    "path": "/action/parameters/amount_cents",
    "predicate": {"predicate_id": "int_range",
                  "parameters": {"min": 1, "max": 1000000}}
  }],
  "policy_epoch": 43,
  "provenance_root": "sha256:...",
  "context_root": "sha256:...",
  "expires_at": 1800000600
}
\end{lstlisting}

\subsection{Field provenance and typed release}
\begin{lstlisting}
{
  "claim": {
    "path": "/action/parameters/amount_cents",
    "source_id": "source:external:17:2",
    "value_digest": "sha256:..."
  },
  "release": {
    "issuer": "key:release",
    "principal": "did:example:user:17",
    "task_root": "task:finance:17",
    "source_id": "source:external:17:2",
    "source_digest": "sha256:...",
    "target_path": "/action/parameters/amount_cents",
    "value_digest": "sha256:...",
    "predicate": {"predicate_id": "int_range",
                  "parameters": {"min": 1, "max": 20000}},
    "action_operation": "payment.transfer",
    "tool_id": "urn:mcp:finance:payments",
    "expires_at": 1800000300
  }
}
\end{lstlisting}

\subsection{Transformation witness and receipt}
\begin{lstlisting}
{
  "witness": {
    "issuer": "key:directory",
    "relation_id": "alias_resolution",
    "field_path": "/action/destination",
    "before_digest": "sha256:...",
    "after_digest": "sha256:...",
    "component_signer": "key:adapter",
    "contract_id": "contract:adapter:2",
    "parameters": {"alias": "alias:finance:17",
                   "resolved": "bankacct:merchant:17"}
  },
  "receipt": {
    "stage": "adapter",
    "role": "adapter",
    "input_digest": "sha256:...",
    "output_digest": "sha256:...",
    "changed_fields": ["/action/destination", "/representation"],
    "transformation_witness_ids": ["witness:..."]
  }
}
\end{lstlisting}

\subsection{Execution permit}
\begin{lstlisting}
{
  "principal": "did:example:user:17",
  "subject": "spiffe://example/agent/finance/4",
  "task_root": "task:finance:17",
  "grant_id": "grant:finance:17",
  "audience": "spiffe://example/finality/finance",
  "action_digest": "sha256:...",
  "bundle_digest": "sha256:...",
  "policy_id": "policy://example/runtime",
  "policy_epoch": 43,
  "nonce": "nonce:...",
  "idempotency_key": "task:...:sha256:...",
  "one_time": true,
  "expires_at": 1800000060
}
\end{lstlisting}

\section{Reproducibility and Reason Codes}
\label{app:reproducibility}

\subsection{Commands}
\needspace{18\baselineskip}
\begin{lstlisting}
python -m venv .venv
. .venv/bin/activate
python -m pip install -r requirements.txt

python -m pytest -q
python scripts/run_experiments.py --output results
python scripts/make_figures.py \
  --results results --output figures
cp figures/*.pdf paper/figures/
cd paper
latexmk -pdf -interaction=nonstopmode \
  -halt-on-error main.tex
\end{lstlisting}

Quick smoke mode uses:
\begin{lstlisting}
python scripts/run_experiments.py \
  --quick --output results_quick
\end{lstlisting}

The full main suite is deterministic. Latency depends on host, Python, and cryptographic-library versions. The manifest records scenario and run counts; raw CSV files expose terminal states, effect counts, first/second outcomes, reason strings, and proof sizes.

\subsection{Expected checks}

A clean artifact run should report 30 passing tests. For the full system, \code{summary.csv} should contain 2,560 attack instances, 128 contained fault--domain classes, effect ASR 0, benign auto-completion 1, ambiguous escalation 1, and lifecycle correctness 1. The exact timing columns are not invariant.

\needspace{27\baselineskip}
\subsection{Representative reason-code families}

\begin{table}[H]
\centering
\scriptsize
\begin{tabularx}{\columnwidth}{@{}p{0.43\columnwidth}X@{}}
\toprule
Reason family & Meaning \\
\midrule
\code{E\_UNTRUSTED\_ROOT} & Chain origin not in trusted ingress set \\
\code{E\_ROOT\_*\_EXCEEDED} & Root authority, scope, class, or field bound exceeded \\
\code{E\_PROVENANCE\_VALUE\_MISMATCH} & Resolved leaf differs from signed manifest claim \\
\code{E\_UNRELEASED\_FIELD} & Untrusted protected field needs validation/escalation \\
\code{E\_INVALID\_RELEASE} & Source, value, predicate, task, tool, or expiry check failed \\
\code{E\_UNAUTHORISED\_STAGE\_SIGNER} & Key is not configured for stage/contract \\
\code{E\_RECEIPT\_PRODUCER\_MISMATCH} & Receipt signer differs from output producer \\
\code{E\_MISSING\_TRANSFORM\_WITNESS} & Changed path lacks required relation evidence \\
\code{E\_TRANSFORM\_*} & Witness role, field, digest, relation, or subject invalid \\
\code{E\_GUARANTEE\_FALSE} & Contract output postcondition failed \\
\code{E\_ACTION\_SUBSTITUTION} & Sink action differs from permit-bound action \\
\code{E\_SUBJECT\_SUBSTITUTION} & Actual caller differs from permit subject \\
\code{E\_REVOKED\_AT\_FINALITY} & Grant or permit revoked after admission \\
\code{E\_UNMEDIATED\_PATH} & Protected effect attempted without permit \\
\bottomrule
\end{tabularx}
\end{table}

\needspace{18\baselineskip}
\subsection{Claim-to-artifact map}

\begin{table}[H]
\centering
\scriptsize
\begin{tabularx}{\columnwidth}{@{}p{0.35\columnwidth}X@{}}
\toprule
Claim & Artifact evidence \\
\midrule
Root and role enforcement & \code{core.py}; root/role regression tests \\
Leaf provenance & JSON Pointer tests; provenance substitution tests \\
Typed release & predicate/value/expiry tests and 300 utility cases \\
Legitimate transformation & alias witness tests; every normal trace \\
Effect containment & \code{raw\_results.csv}, \code{by\_fault.csv} \\
Ablations & \code{ablation\_raw.csv}, \code{ablation\_summary.csv} \\
Lifecycle correctness & first/second outcomes and effect-count columns \\
Performance & \code{performance.csv}, \code{scaling.csv} \\
\bottomrule
\end{tabularx}
\end{table}

\balance
\printbibliography

\end{document}